\documentclass{amsart}

\usepackage{graphicx}
\usepackage{amsmath}
\usepackage{amssymb}

\newtheorem{theorem}{Theorem}[section]

\newtheorem{cor}{corollary}[section]

\theoremstyle{definition}

\theoremstyle{remark}
\newtheorem{remark}[theorem]{Remark}

\newcommand{\RNum}[1]{\uppercase\expandafter{\romannumeral#1\relax}}

\begin{document}

\title{FRACTIONAL SCHR\"{O}DINGER EQUATION WITH ZERO AND LINEAR POTENTIALS}

\author{Saleh Baqer}
\address{Department of Mathematics, Faculty of Science, Kuwait University}
\curraddr{Department of Mathematics and Statistics,
Case Western Reserve University, Cleveland, Ohio 43403}
\email{mpcosmo57@gmail.com}

\author{Lyubomir Boyadjiev}
\address{Department of Mathematics, Faculty of Science, Kuwait University}
\email{boyadjievl@yahoo.com}


\begin{abstract}
This paper is about the fractional Schr\"{o}dinger equation expressed in terms of the Caputo time-fractional and quantum Riesz-Feller space fractional derivatives for particle moving in a potential field. The cases of free particle (zero potential) and a linear potential are considered. For free particle, the solution is obtained in terms of the Fox $H$-function. For the case of a linear potential, the separation of variables method allows the fractional Schr\"{o}dinger equation to be split into space fractional and time fractional ones. By using the Fourier and Mellin transforms for the space equation and the contour integrals technique for the time equation, the solutions are obtained also in terms of the Fox $H$-function. Moreover, some recent results related to the time fractional equation have been revised and reconsidered. The results obtained in this paper contain as particular cases already known results for the fractional Schr\"{o}dinger equation in terms of the quantum Riesz space-fractional derivative and standard Laplace operator.
\end{abstract}

\maketitle

\section*{Introduction}
The Schr\"{o}dinger equation describes dynamics of the state of any quantum mechanical system in terms of its wave function through the Hamiltonian operator of this system. For the wave function $\psi$ of a quantum particle with a mass $m$ moving in a potential field $V$, the Schr\"{o}dinger equation is of the form
\begin{equation}\label{n1.1}
-\frac{\hslash^{2}}{2m}\Delta\psi(x,t)+V(x,t)\psi(x,t)=i\hslash\frac{\partial{}}{\partial{}t}\psi{}\left(x,t\right),
\end{equation}
where $\Delta=\nabla^{2}=\frac{{\partial{}}^2}{{\partial{}x}^2}$ is the Laplacian operator and $\hslash=\frac{h}{2\pi}$ with $h$ being the Planck constant.
The quantum mechanics including the Schr\"{o}dinger equation can be developed by means of the Feynman integrals \cite{9}, e.g., integration over Brownian-like quantum-mechanical paths. It is well known that the Brownian motion is a special case of the L\'{e}vy $\alpha-$stable random processes discovered independently by L\'{e}vy and Khintchine almost a century ago. When $\alpha=2$, the L\'{e}vy $\alpha-$stable distribution is reduced to the Gaussian probability distribution and thus the L\'{e}vy motion becomes the Brownian motion. 
Recently, it was shown in \cite{13} that the Feynamn path integrals over the L\'{e}vy-like quantum- mechanical paths allow to develop a fractional generalization of the quantum mechanics. In particular, whereas the Feynamn path integrals over Brownian trajectories lead to the classical Schr\"{o}dinger equation (\ref{n1.1}), the path integrals over L\'{e}vy trajectories lead to the fractional Scr\"{o}dinger equation ($1<\alpha\leq{2}$), 
\begin{equation}\label{n1.2}
{C}_{\alpha{}}{(-\Delta{})}^{\frac{\alpha{}}{2}}\psi{}\left(x,t\right)+V\left(x,t\right)\
\psi{}\left(x,t\right)=i\hslash{}\frac{\partial{}}{\partial{}t}\psi{}\left(x,t\right),
\end{equation}
where ${C}_{\alpha{}}$ is a positive constant that equals to $\frac{\hslash^{2}}{2m}$ for $\alpha=2$. The operator ${(-\Delta{})}^{\frac{\alpha{}}{2}}$ is introduced as a pseudo-differential operator with the symbol $\vert{k}\vert^{\alpha}$ \cite{14},
\begin{equation}\label{n1.3}
\mathcal{F}\left\{{(-\Delta{})}^{\frac{\alpha{}}{2}}\psi{}\left(x,t\right);k\right\}=\vert{k}\vert^{\alpha}\hat{\psi}\left(k,t\right),
\end{equation}
where
\[
\hat{\psi}\left(k,t\right)=\mathcal{F}\left\{\psi{}\left(x,t\right);k\right\}=\frac{1}{2\pi{}}\int_{-\infty{}}^{\infty{}}e^{-ikx}\
\psi{}\left(x,t\right)\ dx
\]
is the Fourier transform of the function $\psi(x,t)$. The operator ${(-\Delta{})}^{\frac{\alpha{}}{2}}$ was named in \cite{13} the quantum Riesz-fractional derivative of order $\alpha$, but most frequently it is referred to as the fractional Laplacian \cite{11}. Let us note that the conventional Riesz fractional derivative of order $\alpha$ is introduced similarly to (\ref{n1.3}) as a pseudo-differential operator with the symbol $-{\vert{}k\vert{}}^{\alpha}$ \cite{18} and thus differs from the quantum Riesz fractional derivative by a sign. For \small{$\alpha=2$}, the quantum Riesz fractional derivative becomes the negative Laplace operator $-\Delta$ and the equation (\ref{n1.2}) reduces to the classical Schr\"{o}dinger equation (\ref{n1.1}). The cases of a zero potential (free particle), linear potential, Dirac-delta potential, coulomb potential and the infinite potential well for the fractional Schr\"{o}dinger equation (\ref{n1.2}) were studied in \cite{6} and \cite{13}. The generalization of the equation (\ref{n1.2}) when the time derivative is replaced by the Caputo time fractional derivative (\cite{5},\cite{17}) of order $\beta$,  
\[
{^{C}_0D}_{t}^{\beta{}}f\left(t\right)=\left\{\begin{array}{lc}
\frac{1}{\Gamma(n-\beta{})}\
\int_0^t\frac{f^{\left(n\right)}(\tau{})}{{(t-\tau{})}^{\beta{}+1-n}}\ d\tau{}, & n-1<\beta{}<n,\ n\in{}\mathbb{N} \\
\frac{d^n}{dt^n}f\left(t\right), & \beta{}=n\in{}\mathbb{N}
\end{array} \right.,
\]
was considered in \cite{3}, \cite{7}. 
Recently, the Riesz-Feller space fractional derivative ${_xD}_{\theta}^{\alpha{}}$ of order $\alpha$ and skewness $\theta$ has been shown to be relevant for anomalous diffusion models \cite{16}. For sufficiently well-behaved function $f$, the Riesz-Feller space-fractional derivative of order $\alpha$, $0<\alpha\leq{2}$ and skewness $\theta$, $\vert{\theta}\vert{}\leq{\text{min}\{\alpha,2-\alpha\}}$ is defined as a pseudo-differential operator
\begin{equation}\label{n1.4}
\mathcal{F}\left\{_xD_{\theta{}}^{\alpha{}}f\left(x\right);k\right\}=-{\gamma{}}_{\alpha{}}^{\theta{}}(k)\hat{f}(k),
\end{equation} 
with ${\gamma{}}_{\alpha{}}^{\theta{}}(k)={\left\vert{}k\right\vert{}}^{\alpha{}}e^{iSgn\left(k\right)\theta{}\pi/2}$. The fractional Riesz-Feller derivative $_xD_{0}^{\alpha{}}$ is in fact the Riesz fractional derivative. 
In this paper, we employ the quantum fractional Riesz-Feller derivative $D^{\alpha}_{\theta}$ of order $\alpha$, $0<\alpha\leq{2}$ and skewness $\theta$, $\vert{\theta}\vert{}\leq{\text{min}\{\alpha,2-\alpha\}}$ that is as a pseudo-differential operator with the symbol ${\gamma{}}_{\alpha{}}^{\theta{}}(k)$, 
\begin{equation}\label{n1.5}
\mathcal{F}\left\{D_{\theta{}}^{\alpha{}}f\left(x\right);k\right\}={\gamma{}}_{\alpha{}}^{\theta{}}(k)\hat{f}(k).
\end{equation}
According to (\ref{n1.4}) and (\ref{n1.5}), the quantum Riesz-Feller derivative and the Riesz-Feller derivative differ by a sign. The main focus of the paper is on the fractional Schr\"{o}dinger equation
\begin{equation}\label{n1.6}
{C}_{\alpha{}}D_{\theta{}}^{\alpha{}}\psi{}\left(x,t\right)+V\left(x,t\right)\
\psi{}\left(x,t\right)=(i\hslash{})^{\beta}~{^{C}_0D}_{t}^{\beta{}}\psi{}\left(x,t\right),
\end{equation}
with $i^{\beta}=e^{\frac{i\beta\pi}{2}}$, $0<\beta\leq{1}$, $0<\alpha\leq{2}$, $\vert{\theta}\vert{}\leq{\text{min}\{\alpha,2-\alpha\}}$, for the case of a time-independent potential $V(x,t)=V(x)$.
The equation (\ref{n1.6}) for a free particle (zero potential) is studied in the second section of the paper. By using the Fourier and Mellin transforms, the solution of this problem is obtained in terms of the Fox H-function. For the particular choice of the parameters $\beta=1$ and $\theta=0$, the solution obtained reduces to the well known solution of the fractional Schr\"{o}dinger equation with the quantum Riesz derivative for a free particle \cite{14} and also generalizes the corresponding result obtained in \cite{18}. 
The equation (\ref{n1.6}) with linear potential is considered in the third section of the paper. By the method of the separation of variables equation (\ref{n1.6}) is split into space fractional and a time fractional equations. The solution of the space fractional equation obtained by using the Fourier and Mellin integral transforms, contains as particular cases some results already published in \cite{6} and \cite{12}. For the solution of the time fractional equation, the contour integral technique is employed. The approach adopted enables to analyze and comment a similar result published in \cite{3}.

\section{Cauchy-type problem for a free particle}
Consider the fractional Schr\"{o}dinger equation with the Caputo time fractional derivative of order $\beta~(0<\beta\leq{1})$ and with the quantum Riesz-Feller space fractional derivative of order $\alpha~(1<\alpha<2)$ for a free particle ($V(x,t)=0$) in the form 
\begin{equation}\label{n2.1}
{C}_{\alpha{}}D_{\theta{}}^{\alpha{}}\psi{}\left(x,t\right)=(i\hslash{})^{\beta}~{^{C}_0D}_{t}^{\beta{}}\psi{}\left(x,t\right).
\end{equation}
Let us denote by $S$ the space of rapidly decreasing test functions.
 \begin{theorem}
 Let $1<\alpha<2$, $0<\beta\leq{1}, \vert{\theta}\vert{}\leq{2-\alpha}$ and $f\in{S}$. Then the Cauchy-type problem for the equation (\ref{n2.1}) subject to an initial condition 
 \begin{equation}\label{n2.2}
 \psi(x,0)=f(x),~x\in{\mathbb{R}}
 \end{equation}
and boundary conditions
\[
\psi(x,t)\rightarrow{0}~\text{as}~x\rightarrow{\pm{\infty}}
\]
is solvable and its unique solution is given by
\begin{equation}\label{n2.3}
\psi(x,t)=\int^{\infty}_{-\infty} G(x-\xi,t) f(\xi)d\xi,
\end{equation}
where  for $x\neq{0}$ the Green function $G$ is represented in terms of the Fox $H-$function 
\begin{equation}\label{n2.4}
G(x,t)=\frac{1}{\alpha\vert{x}\vert{}}H_{2,2}^{1,1}\left[{\left(\frac{\hslash^{\beta{}}}{C_{\alpha{}}t}\right)}^{\frac{1}{\alpha{}}}e^{-\frac{i\pi{}\left(2-\beta{}\right)}{2\alpha{}}}\
\left\vert{x}\vert{}\right.\left|\begin{array}{
cc}
(1,\frac{1}{\alpha}),(1,\tau) \\
(1,1),(1,\tau)
\end{array}\right.\right]
\end{equation}
with $\tau=\frac{\alpha-sign(x)\theta}{2\alpha}$ and 
\[
G(0,t)=\frac{1}{\pi\alpha}{\left(\frac{\hslash^{\beta{}}}{C_{\alpha{}}t}\right)}^{\frac{1}{\alpha{}}}\Gamma\left(\frac{1}{\alpha}\right)\cos{\left(\frac{\pi{}\theta{}}{2\alpha{}}\right)} e^{-\frac{i\pi{}\left(2-\beta{}\right)}{2\alpha{}}}.
\]
\end{theorem}

\begin{proof}
Denote the Laplace transform $\mathcal{L}\left\{\psi(x,t);s\right\}$ of $\psi$ with respect to the time variable $t$ by $\tilde{\psi}(x,s)$ and the Fourier transform $\mathcal{F}\left\{\psi(x,t);k\right\}$ of $\psi$ with respect to the spatial variable $x$ by $\hat{\psi}(k,t)$. The application of the Laplace transform followed by the Fourier transform to the equation (\ref{n2.1}) leads by (\ref{n1.5}) and (\ref{n2.2}) to the formula 
\begin{equation}\label{n2.5}
\hat{\tilde{\psi{}}}\left(k,s\right)=\frac{\hat{f}(k)}{s+{\eta{}}_{\alpha{},\beta{}}^{\theta{}}(k)}
\end{equation} 
for the Laplace-Fourier transform of the solution $\psi$, where 
\[
{\eta{}}_{\alpha{},\beta{}}^{\theta{}}\left(k\right)=-\frac{C_{\alpha{}}}{h^{\beta{}}}{\left\vert{}k\right\vert{}}^{\alpha{}}e^{-\frac{i\pi{}}{2\alpha{}}\left[\beta{}-sign(k)\theta{}\right]}.
\]
The representation (\ref{n2.3}) then follows from (\ref{n2.5}) by using the inverse Laplace transform and the convolution theorem for the Fourier transform. 
The Green function $G(x,t)$ (fundamental solution) of the problem corresponds to the initial condition $f(x)=\delta(x)$ with the Dirac $\delta-$function and thus 
\begin{equation}\label{n2.6}
\hat{\tilde{G}}\left(k,s\right)=\frac{1}{s+{\eta{}}_{\alpha{},\beta{}}^{\theta{}}(k)}.
\end{equation}
By means of the scaling rules for the Fourier and for the Laplace transforms \cite{8}
\[
\mathcal{F}\left\{f\left(ax\right);k\right\}=\frac{1}{a}\hat{f}\left(\frac{k}{a}\right),a>0,~~\mathcal{L}\left\{f\left(bt\right);s\right\}=\frac{1}{b}\tilde{f}\left(\frac{s}{b}\right),b>0, 
\]
the following scaling property of the Green function can be deduced from (\ref{n2.6}):
\begin{equation}\label{n2.7}
G(x,t)=t^{-\frac{1}{\alpha}}{K}_{\alpha{},\beta{}}^{\theta{}}\left(\frac{x}{t^{\frac{1}{\alpha}}}\right),
\end{equation}
where the single variable function ${K}_{\alpha{},\beta{}}^{\theta{}}$ is to be determined. The application of the inverse Laplace transform to (\ref{n2.6}) implies \cite{16}
\begin{equation}\label{n2.8}
\hat{G}\left(k,t\right)=e^{-{\eta{}}_{\alpha{},\beta{}}^{\theta{}}(k)t},~~k\in{\mathbb{R}},t\geq{0}.
\end{equation}
Due to the symmetry property 
\[
{\eta{}}_{\alpha{},\beta{}}^{-\theta{}}\left(-k\right)={\eta{}}_{\alpha{},\beta{}}^{\theta{}}(k),
\]
the relation 
\begin{equation}\label{n2.9}
K_{\alpha{},\beta{}}^{\theta{}}\left(-x\right)=K_{\alpha{},\beta{}}^{-\theta{}}(x)
\end{equation}
holds and hence we can restrict our further considerations to the case $x\geq{0}$. 
Let us first consider the case $x>0$. Then from (\ref{n2.7}) and (\ref{n2.8}) it follows that 
\begin{multline*}
K_{\alpha{},\beta{}}^{\theta{}}\left(x\right)=G\left(x,1\right)=\frac{1}{2\pi{}}\int_{-\infty{}}^{\infty{}}e^{-ikx}e^{-{\eta{}}_{\alpha{},\beta{}}^{\theta{}}(k)}\
dk\
\\=\frac{1}{2\pi{}}\left\{\int_0^{\infty{}}e^{-ikx}e^{-{\eta}_{\alpha{},\beta{}}^{\theta{}}(k)}dk+\int_0^{\infty{}}e^{ikx}e^{-{\eta}_{\alpha{},\beta{}}^{-\theta{}}(k)}dk\right\}.
\end{multline*}
Thus, we see that 
\begin{equation}\label{n2.10}
K_{\alpha{},\beta{}}^{\theta{}}\left(x\right)=\frac{1}{2}\left[{{}_cK}_{\alpha{},\beta{}}^{\theta{}}\left(x\right)+{{}_sK}_{\alpha{},\beta{}}^{\theta{}}\left(x\right)\right],
\end{equation}
where
\[
{{}_cK}_{\alpha{},\beta{}}^{\theta{}}\left(x\right)=\frac{1}{\pi{}}\int_0^{\infty{}}\cos{(kx)}\left[e^{-{\eta}_{\alpha{},\beta{}}^{\theta{}}(k)}+e^{-{\eta}_{\alpha{},\beta{}}^{-\theta{}}(k)}\right]dk,
\]
\[
{{}_sK}_{\alpha{},\beta{}}^{\theta{}}\left(x\right)=\frac{-i}{\pi{}}\int_0^{\infty{}}\sin{(kx)}\left[e^{-{\eta{}}_{\alpha{},\beta{}}^{\theta{}}(k)}-e^{-{\eta{}}_{\alpha{},\beta{}}^{-\theta{}}(k)}\right]dk.
\]
To determine the functions ${{}_cK}_{\alpha{},\beta{}}^{\theta{}}$ and ${{}_sK}_{\alpha{},\beta{}}^{\theta{}}$, we employ technique of the Mellin integral transform. For a sufficiently well-behaved function $f$, the Mellin transform is defined by the formula 
\[
\mathcal{M}\left\{f\left(r\right);s\right\}=f^*\left(s\right)=\int_0^{\infty{}}f\left(r\right)\
r^{s-1}dr,~\gamma_{1}<\text{Re}(s)<\gamma_{2},
\]
and the inverse Mellin transform is given by 
\[
\mathcal{M}^{-1}\left\{f^*\left(s\right);r\right\}=f\left(r\right)=\frac{1}{2\pi{}i}\int_{\gamma{}-i\infty{}}^{\gamma{}+i\infty{}}{f^*\left(s\right)\,
r}^{-s}ds
\]
where $r>0,\gamma=\text{Re}(s),\gamma_{1}<\gamma<\gamma_{2}.$
Proceeding as in \cite{16}, by the Mellin convolution rule we get the representations 
\[
{{}_cK}_{\alpha{},\beta{}}^{\theta{}}\left(x\right)=\frac{1}{\pi{}x}\frac{1}{2\pi{}i}\int_{\gamma{}-i\infty{}}^{\gamma{}+i\infty{}}{\left[e^{-{\eta{}}_{\alpha{},\beta{}}^{\theta{}}\left(k\right)}+e^{-{\eta{}}_{\alpha{},\beta{}}^{-\theta{}}\left(k\right)}\right]}^*(s)\
\Gamma\left(1-s\right)\sin{\left(\frac{\pi{}s}{2}\right)x^s}ds,
\]
for $x>0$, $0<\gamma<1$ and 
\[
{{}_sK}_{\alpha{},\beta{}}^{\theta{}}\left(x\right)=\frac{-1}{\pi{}x}\frac{1}{2\pi{}}\int_{\gamma{}-i\infty{}}^{\gamma{}+i\infty{}}{\left[e^{-{\eta{}}_{\alpha{},\beta{}}^{\theta{}}\left(k\right)}-e^{-{\eta{}}_{\alpha{},\beta{}}^{-\theta{}}\left(k\right)}\right]}^*(s)\
\Gamma\left(1-s\right)\cos{\left(\frac{\pi{}s}{2}\right)x^s}ds,
\]
for $x>0$ and $0<\gamma<2$. The formula (\cite{16}, (6.1)) for the Mellin transform of the Mittag-Leffler function yields 
\[
{\left[e^{-{\eta{}}_{\alpha{},\beta{}}^{\theta{}}\left(k\right)}+e^{-{\eta{}}_{\alpha{},\beta{}}^{-\theta{}}\left(k\right)}\right]}^*\left(s\right)=\frac{2}{\alpha{}}\Gamma\left(\frac{s}{\alpha{}}\right){\left[\frac{-C_{\alpha{}}}{{\left(ih\right)}^{\beta{}}}\right]}^{\frac{-s}{\alpha{}}}\cos{\left(\frac{\pi{}\theta{}s}{2\alpha{}}\right)},
\]
\[
{\left[e^{-{\eta{}}_{\alpha{},\beta{}}^{\theta{}}\left(k\right)}-e^{-{\eta{}}_{\alpha{},\beta{}}^{-\theta{}}\left(k\right)}\right]}^*\left(s\right)=\frac{-2i}{\alpha{}}\Gamma\left(\frac{s}{\alpha{}}\right){\left[\frac{-C_{\alpha{}}}{{\left(ih\right)}^{\beta{}}}\right]}^{\frac{-s}{\alpha{}}}\sin{\left(\frac{\pi{}\theta{}s}{2\alpha{}}\right)},
\]
and hence
\[
\begin{split}
{{}_cK}_{\alpha{},\beta{}}^{\theta{}}\left(x\right)=\frac{1}{\alpha{}\pi{}x}\frac{1}{\pi{}i}\int_{\gamma{}-i\infty{}}^{\gamma{}+i\infty{}}\
\Gamma\left(\frac{s}{\alpha{}}\right)\Gamma\left(1-s\right)\times\\\cos{\left(\frac{\pi{}\theta{}s}{2\alpha{}}\right)}\
&\sin{\left(\frac{\pi{}s}{2}\right)}{\left[{\left(\frac{\hslash^{\beta{}}}{C_{\alpha{}}}\right)}^{\frac{1}{\alpha{}}}e^{-\frac{i\pi{}(2-\beta{})}{2\alpha{}}}x\right]}^s ds,
\end{split}
\]

\[
\begin{split}
{{}_sK}_{\alpha{},\beta{}}^{\theta{}}\left(x\right)=\frac{-1}{\alpha{}\pi{}x}\frac{1}{\pi{}i}\int_{\gamma{}-i\infty{}}^{\gamma{}+i\infty{}}\
\Gamma\left(\frac{s}{\alpha{}}\right)\Gamma\left(1-s\right)\times\\\sin{\left(\frac{\pi{}\theta{}s}{2\alpha{}}\right)}\
&\cos{\left(\frac{\pi{}s}{2}\right)}{\left[{\left(\frac{\hslash^{\beta{}}}{C_{\alpha{}}}\right)}^{\frac{1}{\alpha{}}}e^{-\frac{i\pi{}(2-\beta{})}{2\alpha{}}}x\right]}^s ds.
\end{split}
\]
Then from (\ref{n2.10}) it follows that
\[
\begin{split}
K_{\alpha{},\beta{}}^{\theta{}}\left(x\right)=\frac{1}{\alpha{}\pi{}x}\frac{1}{2\pi{}i}\int_{\gamma{}-i\infty{}}^{\gamma{}+i\infty{}}\
\Gamma\left(\frac{s}{\alpha{}}\right)\Gamma\left(1-s\right)\times\\\sin{\left(\frac{\left(\alpha{}-\theta{}\right)\pi{}s}{2\alpha{}}\right)}
&{\left[{\left(\frac{\hslash^{\beta{}}}{C_{\alpha{}}}\right)}^{\frac{1}{\alpha{}}}e^{-\frac{i\pi{}(2-\beta{})}{2\alpha{}}}x\right]}^s{ds}.
\end{split}
\]
The substitution $\tau_{1}=\frac{\alpha-\theta}{2\alpha}$ and the reflection formula 
\[
\Gamma(\tau_{1}s)\Gamma(1-\tau_{1}s)=\frac{\pi}{\sin(\tau_{1}\pi s)}
\]
lead to the representation 
\[
K_{\alpha{},\beta{}}^{\theta{}}\left(x\right)=\frac{1}{\alpha{}x}\frac{1}{2\pi{}i}\int_{\gamma{}-i\infty{}}^{\gamma{}+i\infty{}}\
\frac{\Gamma\left(\frac{s}{\alpha{}}\right)\Gamma\left(1-s\right)}{\Gamma\left({\tau{}}_1s\right)\Gamma\left(1-{\tau{}}_1s\right)}{\left[{\left(\frac{\hslash^{\beta{}}}{C_{\alpha{}}}\right)}^{\frac{1}{\alpha{}}}e^{-\frac{i\pi{}(2-\beta{})}{2\alpha{}}}x\right]}^sds.
\]
The last integral can be represented in terms of the Fox $H$-function (\cite{15}) as follows 
\begin{equation}\label{n2.11}
K_{\alpha{},\beta{}}^{\theta{}}\left(x\right)=\frac{1}{\alpha x}H_{2,2}^{1,1}\left[{\left(\frac{C_{\alpha{}}}{\hslash^{\beta{}}}\right)}^{\frac{1}{\alpha{}}}e^{\frac{i\pi{}\left(2-\beta{}\right)}{2\alpha{}}}
\frac{1}{x}\left|\begin{array}{
cc}
(0,1),(0,\tau_{1}) \\
(0,\frac{1}{\alpha}),(0,\tau_{1})
\end{array}\right.\right].
\end{equation}
By means of the reciprocal formula for the $H$-function (\cite{15}, (1.58)) and from (\ref{n2.11}) we get 
\[
K_{\alpha{},\beta{}}^{\theta{}}\left(x\right)=\frac{1}{\alpha x}H_{2,2}^{1,1}\left[{\left(\frac{\hslash^{\beta{}}}{C_{\alpha{}}}\right)}^{\frac{1}{\alpha{}}}e^{-\frac{i\pi{}\left(2-\beta{}\right)}{2\alpha{}}}\
x\left|\begin{array}{
cc}
(1,\frac{1}{\alpha}),(1,\tau_{1}) \\
(1,1),(1,\tau_{1})
\end{array}\right.\right],
\]
and thus according to (\ref{n2.7}) the Green function $G(x,t)$ takes the form
\begin{equation}\label{n2.12}
G(x,t)=\frac{1}{\alpha x}H_{2,2}^{1,1}\left[{\left(\frac{\hslash^{\beta{}}}{tC_{\alpha{}}}\right)}^{\frac{1}{\alpha{}}}e^{-\frac{i\pi{}\left(2-\beta{}\right)}{2\alpha{}}}\
x\left|\begin{array}{
cc}
(1,\frac{1}{\alpha}),(1,\tau_{1}) \\
(1,1),(1,\tau_{1})
\end{array}\right.\right].
\end{equation}
In the case $x<0$, we use the formulas (\ref{n2.9}) and (\ref{n2.12}) to obtain the representation 
\[
G(x,t)=\frac{1}{\alpha (-x)}H_{2,2}^{1,1}\left[{\left(\frac{\hslash^{\beta{}}}{tC_{\alpha{}}}\right)}^{\frac{1}{\alpha{}}}e^{-\frac{i\pi{}\left(2-\beta{}\right)}{2\alpha{}}}\
(-x)\left|\begin{array}{
cc}
(1,\frac{1}{\alpha}),(1,\tau_{2}) \\
(1,1),(1,\tau_{2})
\end{array}\right.\right],
\]
with $\tau_{2}=\frac{\alpha+\theta}{2\alpha}$.
Finally, in the case $x=0$, the representation (\ref{n2.10}) takes the form 
\[
K_{\alpha{},\beta{}}^{\theta{}}(0)=\frac{1}{2\pi}\int^{\infty}_{0}{\left[e^{-{\eta{}}_{\alpha{},\beta{}}^{\theta{}}\left(k\right)}+e^{-{\eta{}}_{\alpha{},\beta{}}^{-\theta{}}\left(k\right)}\right]}dk.
\]
The right-hand side of the last formula can be interpreted as the Mellin transform of the function $\frac{1}{2\pi}\left[e^{-{\eta{}}_{\alpha{},\beta{}}^{\theta{}}\left(k\right)}+e^{-{\eta{}}_{\alpha{},\beta{}}^{-\theta{}}\left(k\right)}\right]$ at the point $s=1$ and thus 
\[
K_{\alpha{},\beta{}}^{\theta{}}(0)=\frac{1}{\alpha\pi}{\left(\frac{\hslash^{\beta{}}}{C_{\alpha{}}}\right)}^{\frac{1}{\alpha{}}}\Gamma\left(\frac{1}{\alpha}\right)\cos\left(\frac{\pi\theta}{2\alpha}\right)e^{-\frac{i\pi{}\left(2-\beta{}\right)}{2\alpha{}}}.
\]
Hence, by (\ref{n2.7}) we get 
\[
G(0,t)=\frac{1}{\alpha\pi}{\left(\frac{\hslash^{\beta{}}}{tC_{\alpha{}}}\right)}^{\frac{1}{\alpha{}}}\Gamma\left(\frac{1}{\alpha}\right)\cos\left(\frac{\pi\theta}{2\alpha}\right)e^{-\frac{i\pi{}\left(2-\beta{}\right)}{2\alpha{}}}
\]
that completes the proof of the theorem.
\end{proof}

\begin{remark}
In the case $\beta=1$ and $\theta=0$, the Schr\"{o}dinger equation (\ref{n2.1}) reduces to the fractional Schr\"{o}dinger equation with the quantum Riesz derivative that has been already studied. In particular, the fractional Schr\"{o}dinger equation with the quantum Riesz derivative was solved for a free particle and its fundamental solution was obtained in the form \cite{14}
\begin{equation}\label{wazzup}
G(x,t)=\frac{1}{\alpha \vert{x}\vert{}}H_{2,2}^{1,1}\left[{\left(\frac{\hslash}{itC_{\alpha{}}}\right)}^{\frac{1}{\alpha{}}}\
\vert{x\vert}\left|\begin{array}{
cc}
(1,\frac{1}{\alpha}),(1,\frac{1}{2}) \\
(1,1),(1,\frac{1}{2})
\end{array}\right.\right]
\end{equation}
that evidently appears as a particular case of (\ref{n2.4}) when taking into consideration that $-i=e^{\frac{-i\pi}{2}}$.
\end{remark}

\begin{remark}
In the case $\alpha=2$, $\beta=1$ and $\theta=0$, the equation (\ref{wazzup}) reduces to
\[
G(x,t)=\frac{1}{2\vert{x}\vert{}}\frac{1}{2\pi{}i}\int_{\gamma{}-i\infty{}}^{\gamma{}+i\infty{}}\
\frac{\Gamma(1-s)}{\Gamma\left(1-\frac{s}{2}\right)} \xi^{s} ds=\frac{1}{2}M_{\frac{1}{2}}\left(\xi\right),
\]
where $\xi={\left(\frac{\hslash}{itC_{2}}\right)}^{\frac{1}{2}}\vert{x\vert}$ and the function $M_{\nu}(z)$ is defined for any order $\nu, 0<\nu<1$, and $z\in{\mathbb{C}}$ as 
\[
M_{\nu}(z)=\sum_{n=0}^{\infty}\frac{(-z)^{n}}{n!\Gamma(-\nu n+(1-\nu))}=\frac{1}{\pi}\sum_{n=1}^{\infty}\frac{(-z)^{n-1}}{(n-1)!} \Gamma(\nu n)\sin(\nu n \pi).
\]
As a special case of the Wright function, $M_{\nu}(z)$ is an entire function of order $\rho=\frac{1}{1-\nu}$ and provides a generalization of the Gaussian and the Airy functions. In particular, 
\[
M_{\frac{1}{2}}(z)=\frac{1}{\sqrt{\pi}}\text{exp}\left(-\frac{z^{2}}{4}\right),~~M_{\frac{1}{3}}=3^{\frac{2}{3}}\text{Ai}\left(\frac{z}{\sqrt[3]{3}}\right). 
\]
Therefore, we have
\[
G(x,t)=\frac{1}{2\sqrt{\pi}}\text{exp}\left(-\frac{\xi^{2}}{4}\right).
\]
\end{remark}

\section{Linear Potential}
Consider a particle in a linear potential field
\[
V(x)=\left\{\begin{array}{ll}
Ax, & x\geq{0}\left(A>0\right) \\
\infty, & x<0
\end{array} \right..
\]
Then by assuming that $\psi(x,t)=f(t)\phi(t)$ and the method of separation of variables, the equation (\ref{n1.6}) reduces to the space fractional equation 
\begin{equation}\label{n3.1}
C_{\alpha{}}D_{\theta{}}^{\alpha{}}\phi(x)+Ax\phi(x)=E\phi(x),~x\geq{}0\ .
\end{equation} 
and the time fractional equation 
\begin{equation}\label{n3.2}
\left(i\hslash\right)^{\beta}{^{C}_0D}^{\beta}_{t}f(t)=Ef(t),
\end{equation} 
where $E$ refers to the energy. 

\subsection{Space fractional equation}\label{subsec:3.1}
\begin{theorem}
If $ 1<\alpha\leq{2} $ and $ \left| \theta \right| \leq{2-\alpha} $, then the solution of the equation (\ref{n3.1})) has the form
\begin{multline*}
\,\,\,\phi(x)=\\
\frac{2\pi{}N}{\left(\alpha{}+1\right)}H_{2,2}^{1,1}\left[\left(x-\frac{E}{A}\right)\frac{1}{\hslash{}}{\left(\frac{C_{\alpha{}}}{\hslash{}A\left(\alpha{}+1\right)}\right)}^{\frac{-1}{\alpha{}+1}}\left|\begin{array}{
ll}
\left(\frac{\alpha{}}{\alpha{}+1},\frac{1}{1+\alpha{}}\right),\left(\frac{2+\alpha{}-\theta{}}{2\left(\alpha{}+1\right)},\frac{\alpha{}+\theta{}}{2\left(\alpha{}+1\right)}\right)
\\
\left(0,1\right),\left(\frac{2+\alpha{}-\theta{}}{2\left(\alpha{}+1\right)},\frac{\alpha{}+\theta{}}{2\left(\alpha{}+1\right)}\right)
\end{array}\right.\right],
\end{multline*}
where
\[
N=\frac{1}{2\pi{}\hslash{}}{\left(\frac{C_{\alpha{}}}{A\hslash{}\left(\alpha{}+1\right)}\right)}^{\frac{-1}{(\alpha{}+1)}}.
\]
\end{theorem}

\begin{proof}
According to (\ref{n1.5}) and the formula
\[
\mathcal{F}\
\left\{x\phi(x);p\right\}=i\hslash{}\frac{d}{dp}\hat{\phi}(p),
\]
the application of the Fourier transform in momentum representation, 
\[
\hat{\phi}(p)=\frac{1}{2\pi{}\hslash{}}\int_{-\infty{}}^{\infty{}}e^{-ipx/\hslash{}}\
\phi(p)\ dx ~~\text{with}~~
\phi(x)=\frac{1}{2\pi\hslash{}}\int_{-\infty{}}^{\infty{}}e^{ipx/\hslash{}}\
\hat{\phi}(p)\ dp,
\]
to the equation (\ref{n3.1}) leads to
\[
\frac{d\hat{\phi}(p)}{\hat{\phi}(p)}=\frac{1}{Ai\hslash{}}\left(E-C_{\alpha{}}{\left\vert{}p\right\vert{}}^{\alpha{}}e^{iSgn(p)\theta{}\pi/2}\right),
\]
from where it follows readily that (omitting the constant of the integration)
\[
\hat{\phi}(p)=\left\{\begin{array}{l}\begin{array}{
ll}
\exp\left[\frac{-i}{A\hslash{}}\left(Ep-\frac{C_{\alpha{}}}{\alpha{}+1}p^{\alpha{}+1}e^{i\theta{}\pi/2}\right)\
\right]; & p>0
\end{array} \\
\begin{array}{
ll}
\exp\left[\frac{-i}{A\hslash{}}\left(Ep+\frac{C_{\alpha{}}}{\alpha{}+1}{\left\vert{}p\right\vert{}}^{\alpha{}+1}e^{-i\theta{}\pi/2}\right)\
\right]; &  p<0
\end{array}\end{array}.\right.
\]
Setting
\begin{equation}\label{n3.3}
w=p{\left(\frac{C_{\alpha{}}}{A\hslash{}\left(\alpha{}+1\right)}\right)}^{\frac{1}{(\alpha{}+1)}},~~y=\frac{1}{\hslash{}}\left(x-\frac{E}{A}\right){\left(\frac{C_{\alpha{}}}{A\hslash{}\left(\alpha{}+1\right)}\right)}^{\frac{-1}{(\alpha{}+1)}},
\end{equation}
it is possible by the application of the inverse Fourier transform to get that
\begin{equation}\label{n3.4}
\phi(y)=N\left\{{\phi{}}_1(y)+{\phi{}}_2(y)\right\},
\end{equation}
where
\[
{\phi}_1(y)=\int_0^{\infty{}}e^{iyw}e^{i{e^{i\theta\pi/2}w}^{\alpha{}+1}}\
\ dw,
\]
and
\[
{\phi}_2(y)=\int_{-\infty{}}^0\
e^{iyw}e^{-i{e^{-i\theta\pi/2}\left\vert{}w\right\vert{}}^{\alpha{}+1}}\
dw.
\]
Denote by $ \hat{\phi}(s)=\mathcal{M}\left\{\phi(y); s\right\} $ the Mellin transform of $ \phi(y) $. If $\rho\in{\mathbb{C}}$, from the formula (\cite[Ch.8]{8}) it follows 
\begin{equation}\label{n3.5}
\mathcal{M}\left\{e^{-i\rho{x}};s\right\}=\left(i\rho\right)^{-s}\Gamma(s)={\rho}^{-s}\Gamma(s)\left(\cos\left(\frac{\pi{s}}{2}\right)-i\sin\left(\frac{\pi{s}}{2}\right)\right),
\end{equation}
it follows that
\begin{equation}\label{n3.6}
{\tilde{\phi}}_1(s)={\left(-i\right)}^{-s}\Gamma\left(s\right)\int_0^{\infty{}}{e}^{ie^{i\theta\pi/2}w^{\alpha{}+1}}w^{-s}dw,
\end{equation}
and
\begin{equation}\label{n3.7}
{\tilde{\phi}}_2(s)={\left(-i\right)}^{-s}\Gamma\left(s\right)\int_{-\infty{}}^0e^{-ie^{-i\theta\pi/2}{\left\vert{}w\right\vert{}}^{\alpha{}+1}}w^{-s}dw.
\end{equation}

Now by the formula (\ref{n3.5}), the substitution $ u={w}^{\alpha+1} $ in (\ref{n3.6}) and the substitutions $ u=-w $, $ \xi={u}^{\alpha+1} $ in (\ref{n3.7}), it can be seen that
\[
{\tilde{\phi}}_1(s)=\frac{1}{\alpha{}+1}\Gamma\left(\frac{1-s}{1+\alpha{}}\right)\Gamma\left(s\right)\exp\left(\frac{i\pi{}}{2}\left[\frac{1-\theta{}}{\alpha{}+1}+\frac{\left(\alpha{}+\theta{}\right)s}{\alpha{}+1}\right]\right),
\]
and
\[
{\tilde{\phi}}_2(s)=\frac{1}{\alpha{}+1}\Gamma\left(\frac{1-s}{1+\alpha{}}\right)\Gamma\left(s\right)\exp\left(\frac{-i\pi{}}{2}\left[\frac{1-\theta{}}{\alpha{}+1}+\frac{\left(\alpha{}+\theta{}\right)s}{\alpha{}+1}\right]\right).
\]
These representations and the formula
\[
\cos\left(\pi{z}/2\right)=\frac{\pi{}}{\Gamma\left(\frac{1+z}{2}\right)\Gamma\left(\frac{1-z}{2}\right)}
\]
allow from (\ref{n3.4}) the following representation to be obtained 
\[
\tilde{\phi}(s)=\frac{2\pi{}N}{\left(\alpha{}+1\right)}\frac{\Gamma\left(s\right)\Gamma\left(\frac{1-s}{1+\alpha{}}\right)}{\Gamma\left(\frac{\alpha{}+\theta{}-\alpha{}s-\theta{}s}{2\left(\alpha{}+1\right)}\right)\Gamma\left(\frac{2+\alpha{}-\theta{}+\alpha{}s+\theta{}s}{2\left(\alpha{}+1\right)}\right)}.
\]
Therefore
\[
\phi(y)=\frac{2\pi{}N}{\left(\alpha{}+1\right)}\frac{1}{2\pi{}i}\int_{\gamma{}-i\infty{}}^{\gamma{}+i\infty{}}\
\frac{\Gamma\left(s\right)\Gamma\left(\frac{1-s}{1+\alpha{}}\right)}{\Gamma\left(\frac{\alpha{}+\theta{}-\alpha{}s-\theta{}s}{2\left(\alpha{}+1\right)}\right)\Gamma\left(\frac{2+\alpha{}-\theta{}+\alpha{}s+\theta{}s}{2\left(\alpha{}+1\right)}\right)}y^{-s}ds,
\]
and hence (\cite[Sec. 1.2]{15}),
\[
\phi(y)=\frac{2\pi{}N}{\left(\alpha{}+1\right)}H_{2,2}^{1,1}\left(y\left|\begin{array}{
ll}
\left(\frac{\alpha{}}{\alpha{}+1},\frac{1}{1+\alpha{}}\right),\left(\frac{2+\alpha{}-\theta{}}{2\left(\alpha{}+1\right)},\frac{\alpha{}+\theta{}}{2\left(\alpha{}+1\right)}\right)
\\
\left(0,1\right),\left(\frac{2+\alpha{}-\theta{}}{2\left(\alpha{}+1\right)},\frac{\alpha{}+\theta{}}{2\left(\alpha{}+1\right)}\right)
\end{array}\right.\right).
\]
Taking into account the substitutions (\ref{n3.3}), the validity of the theorem follows immediately.
\end{proof}

\begin{cor}(\cite{6})
If $ 1<\alpha\leq{2} $ and $ \theta=0 $, the solution of the equation (\ref{n3.1}) has the form
\begin{multline*}
\,\,\,\phi(x)=\\
\frac{2\pi{}N}{\left(\alpha{}+1\right)}H_{2,2}^{1,1}\left[\left(x-\frac{E}{A}\right)\frac{1}{\hslash{}}{\left(\frac{D_{\alpha{}}}{\hslash{}A\left(\alpha{}+1\right)}\right)}^{\frac{-1}{\alpha{}+1}}\left|\begin{array}{
ll}
\left(\frac{\alpha{}}{\alpha{}+1},\frac{1}{1+\alpha{}}\right),\left(\frac{2+\alpha{}}{2\left(\alpha{}+1\right)},\frac{\alpha{}}{2\left(\alpha{}+1\right)}\right)
\\
\left(0,1\right),\left(\frac{2+\alpha{}}{2\left(\alpha{}+1\right)},\frac{\alpha{}}{2\left(\alpha{}+1\right)}\right)
\end{array}\right.\right],
\end{multline*}
where ${D}_{\alpha{}}$ is the generalized fractional diffusion coefficient with physical dimension $\left[D_{\alpha{}}\right]={erg}^{1-\alpha{}}\times{}{cm}^{\alpha{}}\times{}{sec}^{-\alpha{}}$ ($D_{\alpha{}}$=$\frac{1}{2m}$ for $\alpha{}=2$).
\end{cor}
\begin{cor}(\cite{12})
If $ \alpha=2 $ and $ \theta=0 $, the standard wave function solution of the equation (\ref{n3.1}) has the form
\[
\phi(x)=\frac{\lambda}{\pi{}}\sum_{k=0}^{\infty{}}\Gamma\left(\frac{k+1}{3}\right)\sin{\left(\frac{2\left(k+1\right)\pi{}}{3}\right)}\frac{1}{k!}{\left(3^{\frac{1}{3}}u\right)}^k,
\]
where $ \lambda $ is a constant and $ u=\left(x-\left(\frac{E}{A}\right)\right){\left(\frac{2mA}{{h}^2}\right)}^{\frac{1}{3}}. $
\end{cor}

\subsection{Time fractional equation}\label{subsec:3.2}
In this subsection, we reconsider and revise some results obtained in \cite{3} due to some contradictory computations. For instance, the parameters $A_{j}, j=1,..,p$ and $B_{j}, j=1,...q$ of the Fox $H$-function were taken in the cited paper as negative values, whereas they must be positive \cite{15}. We also noticed a mistake in the sign of one of the Gamma functions arguments in the equation (\cite{3}, (D13)). 
\begin{theorem}
If \,$0<\beta{}<1$ and $E>0$, the solution of the time fractional equation (\ref{n3.2}) is of the form
\[
\
f\left(t\right)=f\left(0\right)\left[\frac{e^{-iE^{\frac{1}{\beta{}}}t/\hslash{}}}{\beta{}}-\frac{1}{\beta{}\rho{}t^{\beta{}}}H_{2,3}^{2,1}\left({\left(-\rho{}\right)}^{\frac{1}{\beta{}}}t\left|\begin{array}{
ll}
\left(1,1/\beta\right),\left(\beta{},1\right) \\
\left(\beta{},1\right),\left(1,1/\beta\right),(\beta{},1)
\end{array}\right)\right.\right],
\]
with $\rho=\left(i\hslash\right)^{\beta}E$. In case $E<0$, the solution of (\ref{n3.2})
\[
f\left(t\right)=-\frac{f\left(0\right)}{\beta{}\rho{}t^{\beta{}}}H_{2,3}^{2,1}\left({\left(-\rho{}\right)}^{\frac{1}{\beta{}}}t\left|\begin{array}{
ll}
\left(1,1/\beta\right),\left(\beta{},1\right) \\
\left(\beta{},1\right),\left(1,1/\beta\right),\left(\beta{},1\right)
\end{array}\right.\right),
\]
in case $0<\beta{}\leq{}\frac{2}{3}$, or
\[
f\left(t\right)=f\left(0\right)\left[\frac{e^{-it{\left\vert{}E\right\vert{}}^{\frac{1}{\beta{}}}e^{\frac{\pi{}i}{\beta{}}}/\hslash{}}}{\beta{}}-\right.\left.\frac{1}{\beta{}\rho{}t^{\beta{}}}H_{2,3}^{2,1}\left({\left(-\rho{}\right)}^{\frac{1}{\beta{}}}t\left|\begin{array}{
ll}
\left(1,1/\beta\right),\left(\beta{},1\right) \\
\left(\beta{},1\right),\left(1,1/\beta\right),\left(\beta{},1\right)
\end{array}\right.\right)\right],
\]
when $\frac{2}{3}<\beta{}<1$.
\end{theorem}

\begin{proof}
Denote by $F(s)$ the Laplace transform $\mathcal{L}\left\{f(t);s\right\}$. Then the application of the Laplace transform to (\ref{n3.2}) implies
\begin{equation}\label{newone}
F(s)=f(0)\frac{s^{\beta-1}}{s^{\beta}-\left(i\hslash\right)^{-\beta}E}.
\end{equation}
Thus,
\[
f\left(t\right)=\mathcal{L}^{-1}\left\{F\left(s\right);t\right\}=\frac{f\left(0\right)}{2\pi{}i}\int_{\gamma{}-i\infty{}}^{\gamma{}+i\infty{}}G\left(s\right)e^{st}\
dt,
\]
where
\[
G\left(s\right)=\frac{s^{\beta{}-1}}{\left(s^{\beta{}}-{\left(i\hslash{}\right)}^{-\beta{}}E\right)}.
\]
The integrand $G\left(s\right)$ in the above integral has a branch point at $s=0$ and at the following
poles:
\[
s_n=E^{\frac{1}{\beta{}}}{\hslash{}}^{-1}e^{i\left(\frac{2\pi{}n}{\beta{}}-\frac{\pi{}}{2}\right)},\
\ \ \ \ n=0,1,2,3,...\ \ .
\]
Thus, the integral can be evaluated by using the contour integration techniques and the residue theorem. For this, it is necessary to choose the Bromwich contour with a branch cut along the negative real axis. This leads to \cite{7}
\[
f(t)=f(0)\left(\sum_nRes\left(G\left(s\right)e^{st};s_n\right)-\mathcal{F}_{\beta{}}\left(\rho{},t\right)\right),
\]
where $\rho{}={\left(i\hslash{}\right)}^{-\beta{}}E$ and

\[
{\ \mathcal{F}}_{\beta{}}\left(\rho{},t\right)=\
\frac{\rho{}\sin{\left(\pi{}\beta{}\right)}}{\pi{}}\int_0^{\infty{}}\left(\frac{x^{\beta{}-1}}{x^{2\beta{}}-2\rho{}\cos{\left(\pi{}\beta{}\right)x^{\beta{}}}+{\rho{}}^2}\right){\
e}^{-xt}dx.
\]
In the domain enclosed by the Bromwich contour chosen, the arguments of the poles should satisfy
\begin{equation}\label{n3.10}
-\pi{}<\arg{\left(s_n\right)}<\pi{},\ \ n=0,1,2,3,...\ \ .
\end{equation}
In case of positive energy ($E>0)$,
the poles $s_n$ are

\[
s_n=E^{\frac{1}{\beta{}}}{\hslash{}}^{-1}e^{i(\frac{2\pi{}n}{\beta{}}-\frac{\pi{}}{2})},\
\ n=0,1,2,3,...\ \ .
\]
By taking into account (\ref{n3.10}), it becomes clear that 
\[
-\pi{}<\frac{2\pi{}n}{\beta{}}-\frac{\pi{}}{2}<\pi,
\]
and thus $n=0$, i.e. $0<\beta<1$. There is only one pole
$s_0=E^{\frac{1}{\beta{}}}{\hslash{}}^{-1}e^{\frac{-i\pi{}}{2}}$ to be considered. It is easy to see that the residue
at $s=s_0$ is
\[
\text{Res}\left(G(s)e^{st};s_0\right)=\frac{e^{-itE^{\frac{1}{\beta{}}}/\hslash{}}}{\beta{}}.
\]
In case of negative energy $(E<0)$,
the poles $s_n$ are
\[
s_n={\left\vert{}E\right\vert{}}^{\frac{1}{\beta{}}}{\hslash{}}^{-1}e^{i(\frac{2\pi{}n}{\beta{}}-\frac{\pi{}}{2}+\frac{\pi{}}{\beta{}})},\
n=0,1,2,3,...\  \  .
\]
According to (\ref{n3.10}),
\[
-\pi{}<\frac{2\pi{}n}{\beta{}}-\frac{\pi{}}{2}+\frac{\pi{}}{\beta{}}<\pi{}\
~~\text{that implies}~~ \frac{-1}{4}\beta{}-\frac{1}{2}<n<\frac{3}{4}\beta{}-\frac{1}{2}.
\]
If $0<\beta{}\leq{}\frac{2}{3}$, then $-\frac{2}{3}<n<0$ and so there are no poles in this case. In case $\frac{2}{3}<\beta{}<1$, $-\frac{3}{4}<n<\frac{1}{4}$ so $n=0$, i.e. there is only one pole $s_0={\left\vert{}E\right\vert{}}^{\frac{1}{\beta{}}}{\hslash{}}^{-1}e^{i(\frac{\pi{}}{\beta{}}-\frac{\pi{}}{2})} $. The residue at $s=s_0$ in this case is:
\[
\text{Res}\left(G(s)e^{st};\
s_0\right)=\frac{e^{-it{\left\vert{}E\right\vert{}}^{\frac{1}{\beta{}}}e^{\frac{\pi{}i}{\beta{}}}/\hslash{}}}{\beta{}}.
\]
Let us recall the Mellin transform property (\cite{8}, (8.3.1))
\begin{equation}\label{n3.11}
\mathcal{M}\left\{f\left(x^{\alpha}\right); s\right\}=\frac{1}{\alpha}f^{*}\left(\frac{s}{\alpha}\right),~\alpha>0,
\end{equation}
and also the following representation valid for $\beta\notin{\mathbb{Z}}$ and $\rho\in{\mathbb{C}}-\left\{0\right\}$ (\cite{2}, (13.126)),
\begin{equation}\label{n3.12}
\mathcal{M}\left\{\frac{1}{x^2-2\rho{}\cos\left(\pi{}\beta{}\right)x+{\rho{}}^2\
};s\right\}=\frac{-\pi{}{\
\left(-\rho{}\right)}^{s-2}\sin(\beta{}\pi{}\left(s-1\right))}{\sin\left(\pi{}s\right)\
\sin(\beta{}\pi{})}.
\end{equation}
Let us denote 
\[
K\left(x\right)=\frac{1}{x^{2\beta{}}-2\rho{}\cos{\left(\pi{}\beta{}\right)x^{\beta{}}+{\rho{}}^2}}.
\]
Then from (\ref{n3.11}), (\ref{n3.12}) and the identity ($z-$not integer)
\[
\sin(\pi{z})=\frac{\pi}{\Gamma(z)\Gamma(1-z)},
\]
it follows that,
\[
\mathcal{M}\left\{K\left(x\right);s\right\}=\frac{\pi{}\
{\left(-\rho{}\right)}^{\frac{s}{\beta{}}}}{\beta{}{\rho{}}^2\sin{\left(\pi{}\beta{}\right)}}\frac{\Gamma\left(s/\beta\right)\Gamma\left(1-\left(s/\beta{}\right)\right)}{\Gamma\left(\beta{}-s\right)\Gamma\left(1-\beta{}+s\right)}.
\]
The application of the inverse Mellin transform results to

\[
\begin{split}
K\left(x\right)&=\frac{\pi{}}{\beta{}{\rho{}}^2\sin{\left(\pi{}\beta{}\right)}}\left[\frac{1}{2\pi{}i}\int_{\gamma{}-i\infty{}}^{\gamma{}+i\infty{}}\frac{\Gamma\left(s/\beta\right)\Gamma\left(1-\left(s/\beta{}\right)\right)}{\Gamma\left(\beta{}-s\right)\Gamma\left(1-\beta{}+s\right)}{\left(\frac{x}{{\left(-\rho{}\right)}^{\frac{1}{\beta{}}}}\right)}^{-s}\
ds\right] \\
 &=\frac{\pi{}}{\beta{}{\rho{}}^2\sin{\left(\pi{}\beta{}\right)}}H_{2,2}^{1,1}\left(\frac{x}{{\left(-\rho{}\right)}^{\frac{1}{\beta{}}}}\left|\begin{array}{
cc}
\left(0,1/\beta\right),(1-\beta{},1) \\
\left(0,1/\beta\right),(1-\beta{},1)
\end{array}\right)\right.,
\end{split}
\]
and hence
\[
\mathcal{L}\left\{K\left(x\right)x^{\beta{}-1};t\right\}=
\frac{\pi{}}{\beta{}{\rho{}}^2\sin{\left(\pi{}\beta{}\right)}}\mathcal{L}\left\{x^{\beta{}-1}\
H_{2,2}^{1,1}\left(\frac{x}{{\left(-\rho{}\right)}^{\frac{1}{\beta{}}}}\left|\begin{array}{
cc}
\left(0,1/\beta\right),(1-\beta{},1) \\
\left(0,1/\beta\right),(1-\beta{},1)
\end{array}\right);t\right\}\right..
\]
According to (\cite{15}, (2.18))

\begin{multline*}
\int_0^{\infty{}}\left(\frac{x^{\beta{}-1}}{x^{2\beta{}}-2\rho{}x^{\beta{}}\cos{\left(\pi{}\beta{}\right)}+{\rho{}}^2}\right){\
e}^{-xt}\
dx\\
=\frac{\pi{}}{\beta{}{\rho{}}^2\sin{\left(\pi{}\beta{}\right)}}t^{-\beta{}}H_{3,2}^{1,2}\left(\frac{t^{-1}}{{\left(-\rho{}\right)}^{\frac{1}{\beta{}}}}\left|\begin{array}{
ll}
\left(1-\beta{},1\right),\left(0,1/\beta\right),(1-\beta{},1) \\
\left(0,1/\beta\right),\left(1-\beta{},1\right)…………
\end{array}\right)\right..
\end{multline*}
By means of the reciprocal formula for the Fox $H$-function (\cite{15}, (1.58)), 

\begin{multline*}
\int_0^{\infty{}}\left(\frac{x^{\beta{}-1}}{x^{2\beta{}}-2\rho{}x^{\beta{}}\cos{\left(\pi{}\beta{}\right)}+{\rho{}}^2}\right){\
e}^{-xt}\
dx\\
=\frac{\pi{}}{\beta{}{\rho{}}^2\sin{\left(\pi{}\beta{}\right)}}t^{-\beta{}}H_{2,3}^{2,1}\left({\left(-\rho{}\right)}^{\frac{1}{\beta{}}}t\left|\begin{array}{
ll}
\left(1,1/\beta\right),\left(\beta{},1\right)…… \\
\left(\beta{},1\right),\left(1,1/\beta\right),(\beta{},1)
\end{array}\right.\right),
\end{multline*}
and hence
\[
{\
\mathcal{F}}_{\beta{}}\left(\rho{},t\right)=\frac{1}{\beta{}\rho{}t^{\beta{}}}H_{2,3}^{2,1}\left({\left(-\rho{}\right)}^{\frac{1}{\beta{}}}t\left|\begin{array}{
ll}
\left(1,1/\beta\right),\left(\beta{},1\right) \\
\left(\beta{},1\right),\left(1,1/\beta\right),(\beta{},1)
\end{array}\right)\right.,
\]
that proves the theorem.
\end{proof}

\begin{remark}
In \cite{3}, the author attempted to show that the solution of the time fractional equation provided by Theorem 3.2. coincides with the well known solution \cite{7}
\begin{equation}\label{thetfs}
f(t)=f(0)E_{\beta}\left[{\left(\frac{t}{i\hslash}\right)}^{\beta}E\right]
\end{equation}
in the case of $0<\beta<1$ and $E>0$. His proof is substantially based on the formula (\cite{3}, (D32))
\begin{equation}\label{false}
\frac{\Gamma\left(-k/\beta\right)}{\Gamma(-k)}=\pm{}{\left(-1\right)}^{k-\left(\frac{k}{\beta{}}\right)}\frac{\Gamma(k+1)}{\Gamma\left(1+\left(k/\beta\right)\right)},~~k=0,1,2,...,
\end{equation}
which is not valid for all $k\in{\mathbb{N}_{0}}$ and $0<\beta<1$. For instance, if $k=1$ and $\beta=\frac{3}{5}$, then the equation (\ref{false}) leads to $0=\pm{1}$. In fact, the discussion whether the solution (\ref{thetfs}) and the result in Theorem 3.2. coincide is needless because of two arguments. The first argument is that the equation (\ref{n3.2}) possesses a unique solution (provided the initial condition $f(0)=f_{0}$ is fixed). Another argument is that the formula (\ref{thetfs}) and the result in Theorem 3.2. both are inverse Laplace transform of (\ref{newone}) and hence they have to coincide. 
\end{remark}


\begin{thebibliography}{99}

\bibitem{1}
L. C. Andrews, \textit{Special Functions of Mathematics for Engineers}.
2nd ed. Spie Press, Washington (1998).

\bibitem{2}
S. S. Bayin, \textit{Mathematical Methods in Science and Engineering}.
Wiley (2006).

\bibitem{3}
S. S. Bayin, Time fractional Schr\"{o}dinger equation: Fox's $H$-functions and the effective potential. \textit{J. Math. Phys}. \textbf{54}, No 1 (2013), 012103; DOI: 10.1063/ 1.4773100. MR3059867. 

\bibitem{4}
S. Baqer, L. Boyadjiev, On the space-time fractional Schr\"{o}dinger equation with time independent potentials. \textit{Contemporary Mathematics}, \textbf{658}, Amer. Math. Soc., Providence, RI (2016), 81--90. 

\bibitem{5}
M. Caputo, Linear models of dissipation whose $Q$ is almost frequency independent. \RNum{2}, \textit{Fract. Calc. Appl. Anal.} \textbf{11}, No 1 (2008), 4--14. Reprinted from Geophys. J. R. Astr. Soc. \textbf{13}, No 5 (1967), 529--539.
 
\bibitem{6}
J. Dong, M. Xu, Some solutions to the space
fractional Schr\"{o}dinger equation using momentum representation method. \textit{J.
Math. Phys}. \textbf{48}, 072105 (2007).

\bibitem{7}
J. Dong, M. Xu. Space-time fractional
Schr\"{o}dinger equation with time-independent potential. \textit{J. Math. Anal. Appl.}
\textbf{344} (2008), 1005--1017.

\bibitem{8}
L. Debnath, D. Bhatta, \textit{Integral Transforms and Their Applications}, Chapman \& Hall/ CRC, Taylor \& Francis Group (2007).

\bibitem{9}
R. P. Feynman, A. R. Hibbs, \textit{Quantum Mechanics and Path Integrals}. McGraw-Hill, New-Hill, New York (1965).

\bibitem{10}
D. J. Griffiths, \textit{Introduction to Quantum Mechanics}. 2nd ed. Prentice-Hall, Englewood Cliffs NJ (2004).

\bibitem{11}
Boling Guo, Zhaohui Huo, Well-posedness for the nonlinear fractional Schr\"{o}dinger equation and inviscid limit behaviour of solution for the fractional Ginzburg-Landau equation. \textit{Fract. Calc. Appl. Anal.} \textbf{16}, No 2 (2013), 226--241.
 
\bibitem{12}
L. D. Landau, E. M. Lifshitz, \textit{A Shorter Course of Theoretical Physics}, \textbf{2}, Pergamon Press, Oxford-New York-Toronto, Ont. (1974). Quantum mechanics, MR0400931(53\# 4761). 

\bibitem{13}
N. Laskin, Fractional quantum mechanics and L\'{e}vy path
integrals. \textit{Phys. Lett. A} \textbf{268} (2000), 298--305.

\bibitem{14}
N. Laskin, Principles of fractional quantum mechanics. Edited by J. Klafter, S. C. Lim and R. Metzler, \textit{Fractional Dynamics: Recent Advances}, World Scientific, Singapore (2012), 393--427.

\bibitem{15}
A. M. Mathai, R. K. Saxena, H. J. Haubold, \textit{The H-Function.
Theory and Applications}. Springer (2010).

\bibitem{16}
F. Mainardi, Yu. Luchko, G. Pagnini, The fundamental solution of the space-time fractional diffusion equation. \textit{Fract. Calc. Appl. Anal.} \textbf{4} (2001), 155--192. 

\bibitem{17}
I. Podlubny, \textit{Fractional Differential Equations}. Academic Press,
New York (1999).

\bibitem{18}
B. Al-Saqabi, L. Boyadjiev, Yu. Luchko, Comments on employing the
Riesz-Feller derivative in the Schr\"{o}dinger equation. \textit{Eur. Phys. J. Special
topics} \textbf{222} (2013), 1779--1794.

\bibitem{19}
S. G. Samko, A. A. Kilbas, O. I. Marichev, \textit{Fractional Integrals and Derivatives: Theory and Applications}. Gordon and Breach (1993). 


\end{thebibliography}
\end{document}